\documentclass[aps, pra, onecolumn, superscriptaddress, nofootinbib, 11pt]{revtex4-2}
\pdfoutput=1
\usepackage[utf8]{inputenc}
\usepackage[english]{babel}
\usepackage[T1]{fontenc}
\usepackage{amsmath,amssymb,amsfonts}
\usepackage{graphicx}
\usepackage{hyperref}
\usepackage[normalem]{ulem}
\usepackage{dcolumn}% Align table columns on decimal point
\usepackage{bm}% bold math
\usepackage{hyperref}% add hypertext capabilities
\usepackage[htt]{hyphenat}
\usepackage{multirow}
\usepackage[table,dvipsnames]{xcolor}
\usepackage{amsmath,amssymb,amsthm, mathtools, enumitem, empheq, gensymb, textcomp, bm, booktabs, multirow,refcount}

\usepackage{hhline}
\usepackage{physics}
\usepackage{subfigure}

\hypersetup{colorlinks,linkcolor={blue},citecolor={blue},urlcolor={blue}}
\newtheorem{theorem}{Theorem}

\newtheorem{lemma}{Lemma}
\newtheorem{definition}{Definition}

\begin{document}
% \preprint{APS/123-QED}

\title{Equivalence between exponential concentration in quantum machine learning kernels and barren plateaus in variational algorithms}

\author{Pranav Kairon}
\affiliation{Department of Physics and Astronomy, University of British Columbia, Vancouver, B.C., Canada}
\affiliation{Stewart Blusson Quantum Matter Institute, Vancouver, B.C., Canada}

\author{Jonas Jäger}
\affiliation{Department of Computer Science and Institute of Applied Mathematics, University of British Columbia, Vancouver, B.C., Canada}
\affiliation{Stewart Blusson Quantum Matter Institute, Vancouver, B.C., Canada}

\author{Roman V. Krems}
\affiliation{Stewart Blusson Quantum Matter Institute, Vancouver, B.C., Canada}
\affiliation{Department of Chemistry, University of British Columbia, Vancouver, B.C., Canada}
\email{rkrems@chem.ubc.ca}

\date{\today}

\begin{abstract}
We show that a parametrized quantum neural network can be used to build a non-parametric quantum kernel for machine learning that inherits the concentration properties of the neural network cost function. This {construction} establishes a rigorous connection between barren plateaus (BP) in variational quantum algorithms and exponential concentration of quantum kernels for machine learning.
Our results imply that recently proposed algorithms for building BP-free quantum circuits can be utilized to construct useful quantum kernels for machine learning with unknown inductive bias.
This is illustrated by a numerical example employing a provably BP-free quantum neural network to construct kernel matrices for multiple classification datasets of increasing dimensionality.
\end{abstract}

\maketitle
\section{Introduction}

Machine learning with quantum kernels \cite{qk_real1,qk_real3,qkreal2_hubregtsen2022training, QK_QNN_same_schuld2021quantum,QML_maria_killoran} and variational quantum algorithms (VQAs) \cite{cerezo_variational_2021} are two distinct paradigms currently explored for computation with noisy quantum devices \cite{nisq_review_1,nisq_review_2,nisq_complexity_2022,cerezo_challenges_2022,biamonte_quantum_2017}. VQAs are restricted by barren plateaus (BPs) \cite{mcclean_barren_2018,larocca2024review}. BPs are regions in the variational parameter space where gradients of the cost function vanish exponentially with the number of qubits, leading to increasingly narrow gorges of optimal solutions \cite{arrasmith2022equivalence}. This implies that an exponentially large number of measurements is required to estimate cost functions for VQAs with BP, as these estimates are otherwise dominated by the statistical measurement error (i.e., shot noise). The origin of BPs can be traced to the vanishing of inner products between quantum states in exponentially large Hilbert spaces. Exponential vanishing of inner products also hinders applications of quantum computing for estimating kernels for machine learning models \cite{ QML_maria_killoran,havlicek_supervised_2019}.
In kernel methods of machine learning, models are based on kernel matrices \cite{havlicek_supervised_2019}, which in quantum machine learning (QML), are estimated by measurements on a quantum computer \cite{qk_real1,qkreal2_hubregtsen2022training,qk_real3}.  Refs. \cite{kubler2021inductive,thanasilp_exponential_2022} show that both the expected value and variance of the
off-diagonal elements of quantum kernel matrices decrease exponentially as the number of qubits used for building quantum kernels grows. This leads to exponential concentration (EC) of the kernel matrix around a diagonal matrix, implying that quantum kernels effectively reduce to $\delta$-functions in a large Hilbert space and making quantum machine learning impractical.

Since BPs preclude the quantum advantage of VQAs \cite{mcclean_barren_2018}, a significant focus of recent work has been on how to engineer quantum circuits for VQAs without BPs.
It has been shown that BPs can be avoided by restricting the depth of quantum circuits
\cite{BP_avoid_shallow}, by building equivariance into the quantum ansatz \cite{equivarint_nguyen_theory_2022}, as well by employing specific initialization strategies for variational optimization \cite{BP_avodi_1, BP_avoid_2, BP_avoid_3, BP_avoid_4, BP_avoid_5}.
Previous work on how to prevent EC of quantum kernels is much more scarce.
The proposed approaches include covariant quantum kernels \cite{glick_covariant_2022}, quantum kernel bandwidth methods \cite{EC_avoid_bandwidth}, quantum Fisher kernels \cite{EC_avoid_fisher_suzuki_quantum_2022}, and projected quantum kernel methods \cite{huang_power_2021}. All these methods encode inductive bias into the quantum models, thereby restricting learning to a particular subspace of the Hilbert space. A rigorously proven relation between BPs in VQAs and concentration of kernel matrices is required to leverage and transfer results from the BP literature to quantum kernels.\looseness=-1

It has been assumed that BPs and EC of quantum kernels are different manifestations of the same problem \cite{wang_noise-induced_2021,express_EC,BP_ent,Uvarov_2021,thanasilp_exponential_2022,larocca2024review}. Thanasilp \textit{et al.} \cite{thanasilp_exponential_2022} rigorously showed how exponential kernel concentration emerges from sources such as expressivity, entanglement, global measurements and hardware noise, which are known sources for BPs. However, Ref.~\cite{thanasilp_exponential_2022}, while limited to a case-by-case study, primarily identifies regimes in which quantum kernels fail (i.e., presence of EC). No rigorous proof of a general\textcolor{black}{ized} connection, or limitations thereof, between BPs in VQAs and EC in quantum kernels has been presented to the best of our knowledge.

\textcolor{black}{Here, we establish the necessary structural bridge between these distinct frameworks by introducing a specific construction that maps between a VQA and a quantum kernel.}
\sout{Here, }\textcolor{black}{Under this explicit mapping,} we formalize the relation between BPs and EC by transferring training bounds derived for VQAs to QML using quantum kernels. We generalize the existing intuition of the emergence of BPs and EC from the same sources \cite{thanasilp_exponential_2022} \textcolor{black}{in}to a bi-directional equivalence ($\mathrm{BP} \iff \mathrm{EC}$) \textcolor{black}{that holds specifically for these induced kernels.} Importantly, this \textcolor{black}{targeted correspondence} allows the transfer of BP-free guarantees to the kernel setting to ensure EC-free models.

There are two possible implications of BPs for quantum kernels. First, training a machine learning model with parametrized kernels, which entails an optimization of kernel parameters, can be directly formulated as a variational problem \cite{qkreal2_hubregtsen2022training,Innan2023}. This optimization problem may suffer from BPs \cite{thanasilp_exponential_2022}. The second implication is EC of quantum kernels. In order to separate these two problems, we consider fixed, non-parametric quantum kernels and examine the effect of BP, or lack thereof, on learning with such quantum kernels constructed from the corresponding quantum neural network.

Schuld  \textit{et al.} \cite{QK_QNN_same_schuld2021quantum} established connections between quantum kernel methods and VQAs. Here, we use a similar approach to demonstrate how training bounds derived for BPs specific to a particular circuit ansatz can, under specific conditions, be extended to overcome EC of quantum kernels. This is achieved by constructing an analogous quantum kernel from the variational ansatz. We then prove that these upper and lower bounds on the cost function can be related to this newly constructed kernel. In particular, we show that if an ansatz leads to BP-free VQAs with a global observable \cite{invariant_schatzki_theoretical_2022}, then there exists such an analogous quantum kernel that is guaranteed not to suffer from EC. We also provide a numerical illustration of Theorem~\ref{theorem1}. Specifically, we construct quantum kernel matrices using the BP-free permutation-equivariant ansatz of Ref.~\cite{invariant_schatzki_theoretical_2022} via Eq.~\eqref{eq: new_kernel}, and benchmark it against the feature map of Havl'i\v{c}ek \textit{et al.} ~\cite{havlicek_supervised_2019}. Across MNIST and additional synthetic classification datasets, we observe that the BP-free construction avoids exponential concentration in the off-diagonal kernel entries, with variance scaling that remains stable under polynomial shot budgets.

\section{\label{sec: qml_theory}Definitions}

For a supervised learning task, the dataset is given by $\mathcal{D} := {(\textbf{x}_i, \textbf{y}_i)}_{i=1}^{N_s}$, where $\textbf{x}_i \in \mathcal{X}$ are input vectors, $\textbf{y}_i \in \mathcal{Y}$ are the labels and $N_s$ is the number of data points.
%Here $\textbf{y}_i$ can be a categorical variable for classification tasks, or a continuous variable for regression problems.
The inputs are related to the labels via a black box function $f: \mathcal{X} \rightarrow \mathcal{Y}$, and the task is to determine a QML model $h_{\boldsymbol{\theta}}$ parameterized by $\boldsymbol{\theta}$ such that $h_{\boldsymbol{\theta}}: \mathcal{X} \rightarrow \mathcal{Y}$ approximates $f$. Note that this work focuses on classical data, unlike quantum data, where the inputs form quantum states \cite{Schuld2021_qml}.

%\subsection{\label{sec:QK} Quantum Kernels}
% \subsubsection{\label{sec:QK}Quantum Kernels}
We consider an $n$-qubit system, and a quantum feature map $\phi$, that encodes $\mathbf{x}_{i}$ into a quantum state in the $n$-qubit Hilbert space, associated with the pure density matrix $\rho(\mathbf{x}_i) = | \phi(\mathbf{x}_i) \rangle \langle \phi(\mathbf{x}_i) |$. The feature map $\phi$ can be implemented through a unitary quantum circuit $U(\mathbf{x})$ acting on a common initial state $\rho_0$. The data embedding $U(\mathbf{x})$ acts as a generator of a unitary ensemble, defined over all possible input data vectors $\mathbf{x}_i \in \mathcal{X}$. From a functional analysis perspective, we have that $U: \mathcal{X} \rightarrow \mathbb{U}_x \subseteq \mathcal{U}(2^n)$, where $\mathcal{U}(2^n)$ is the total space of unitaries and $\mathbb{U}_x = \{ U(\mathbf{x}) | \mathbf{x} \in \mathcal{X} \}$. The density matrices live in the $2^n$-dimensional Hilbert space $\mathcal{H}$, where the inner product is defined as $\bra{\rho}\ket{\rho^\prime}_{\mathcal{H}}=\text{Tr}(\rho\rho^\prime)$, i.e., the Hilbert-Schmidt inner product. An  inner product of quantum states in $\mathcal{H}$ can be be used to define a kernel function $\kappa(\mathbf{x}, \mathbf{x}^{\prime})$ of a reproducing kernel Hilbert space,
\begin{align}\label{eq:qkmain}
    \kappa(\mathbf{x}, \mathbf{x}^{\prime}) = |\langle \phi(\mathbf{x})|\phi(\mathbf{x}^{\prime}) \rangle|^2 = |\langle 0^{\otimes n}\mid U^{\dagger}(\mathbf{x})U(\mathbf{x}^{\prime})\mid 0^{\otimes n}\rangle|^2 = \operatorname{Tr}(\rho(\mathbf{x})\rho(\mathbf{x}^{\prime}))
 %    \\ \rho(\mathbf{x}_i) = U(\mathbf{x}_i) \rho_0 U^\dagger(\mathbf{x}_i)
\end{align}
where $\rho(\mathbf{x}) = U(\mathbf{x}) \rho_0 U^\dagger(\mathbf{x})$ and $\rho(\mathbf{x}^{\prime}) = U(\mathbf{x}^\prime) \rho_0 U^\dagger(\mathbf{x}^\prime)$ are the density matrices corresponding to the quantum states representing the datapoints $\mathbf{x}$ and $\mathbf{x}^{\prime}$ \cite{havlicek_supervised_2019, QML_maria_killoran}, and  $\rho_0=|0\rangle\langle0|^{\otimes^{n}}$. The sequence $U^{\dagger}(\mathbf{x})U(\mathbf{x}^{\prime})$ in Eq.~(\ref{eq:qkmain}) can be used to implement a kernel on a quantum computer, as illustrated by the diagram of the quantum circuit used to evaluate kernel entries in Fig.~\ref{fig: 1}.

%Quantum kernels have been used to solve classically intractable problems \cite{DLP,jager2023universal} thus demonstrating a quantum advantage. They have also been applied to real-world datasets \cite{qk_real1,qkreal2_hubregtsen2022training,qk_real3}. However, the training advantages of quantum kernels relies on an efficient evaluation of the kernel matrix. This condition is  not always easy to fulfill given the constraints of current quantum hardware and the problem of exponential concentration.

\begin{figure}
  \centering
  \includegraphics[width=0.99\columnwidth]{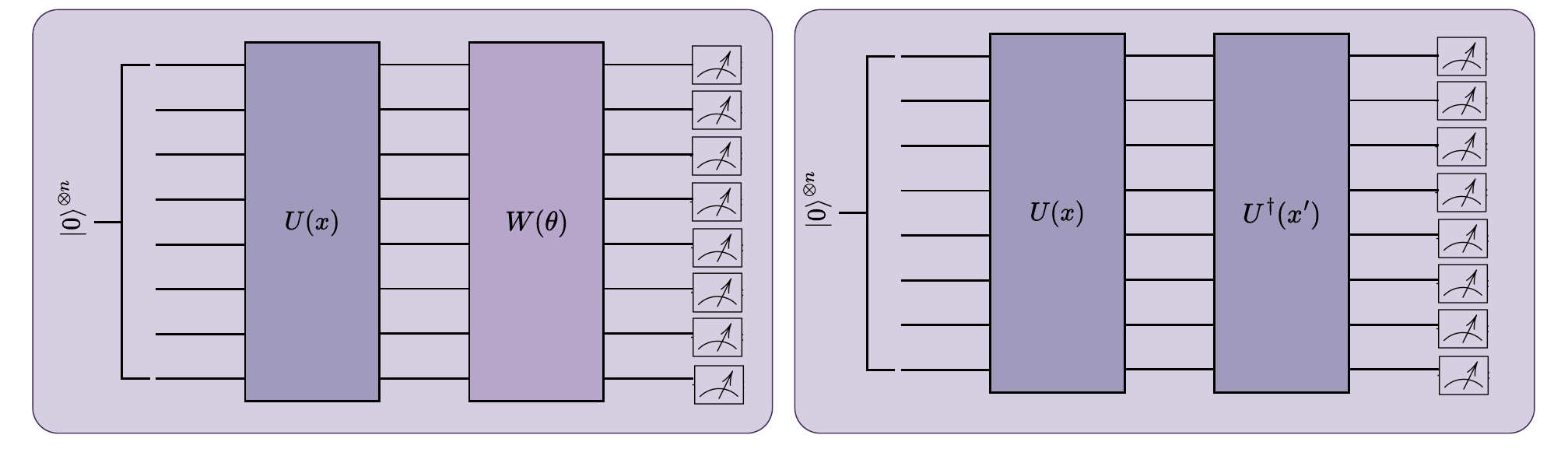}
  \caption{Schematic quantum circuit for a QNN (left) and a quantum kernel (right). The final output is the expectation value of $|0\rangle\langle0|^{\otimes{n}}$. }
  \label{fig: 1}
\end{figure}

%\subsubsection{Exponential Concentration}
EC of quantum kernels is defined as follows.
\begin{definition}
\rm {\bf (Deterministic)} Deterministic exponential concentration of quantum kernel $\kappa(\mathbf{x},\mathbf{x}^{\prime})$ in the number of qubits $n$ is defined as
    \begin{align*}
    & \mid \kappa(\mathbf{x},\mathbf{x}^{\prime}) - \mu \mid \le \beta, \quad \beta \in O(1/b^n)
    \\
    & \mu=\mathbb{E}_{{\mathbf{x}},{\mathbf{x}}^{\prime}}[\kappa(\mathbf{x},\mathbf{x}^{\prime})]
    \end{align*}
for some $b>1$ and all pairs of data points $\mathbf{x},\mathbf{x}^{\prime}$.
\end{definition}
\begin{definition}
\rm {\bf (Probabilistic)}
Probabilistic concentration of quantum kernel $\kappa(\mathbf{x},\mathbf{x}^{\prime})$ is
\begin{equation}\mathrm{Pr}_{\mathbf{x},\mathbf{x}^{\prime}}[ \mid \kappa(\mathbf{x},\mathbf{x}^{\prime}) - \mu \mid \geq \delta ]  \leq  \frac{\beta^2}{\delta^2},\quad \beta \in O(1/b^n)\end{equation} Here we have made use of Chebyshev's inequality:
\begin{equation}
    \Pr(| \kappa(\mathbf{x},\mathbf{x}^{\prime}) - \mu | \geq \delta) \leq \frac{\operatorname{Var}_{\mathbf{x},\mathbf{x}^{\prime}}[\kappa(\mathbf{x},\mathbf{x}^{\prime})]}{\delta^2}
\end{equation}
\end{definition}
Note that probabilistic kernel concentration is more general than deterministic since any deterministically concentrated kernel is also probabilistically concentrated but the reverse is generally not true.
%This implies that the probability that $\kappa(\mathbf{x},\mathbf{x}^{\prime})$ deviates from $\mu$ by a small amount $\delta$ is exponentially suppressed. Additionally if the mean of the kernel matrix also vanishes exponentially in the limit of large number of qubits i.e. $\mu \leq O(1/b^n)$ then we can say that our kernel matrix concentrates exponentially fast towards an exponentially small value.
EC can also be observed in the eigenspectrum of the kernel matrix. The framework proposed in Ref.~\cite{canatar_spectral_2021} considers the quantity of samples necessary for a model to learn a target function. The learning requirement for each eigenmode of the kernel is determined by its associated eigenvalue. Consequently, a ``flat" spectrum indicates a requirement of an exponential amount of data to learn each eigenmode, which leads to a high generalization error. A numerical demonstration of exponential concentration is given in Fig.~\ref{fig: eckernel}.

\begin{figure}
  \centering
  \includegraphics[width=0.95\columnwidth]{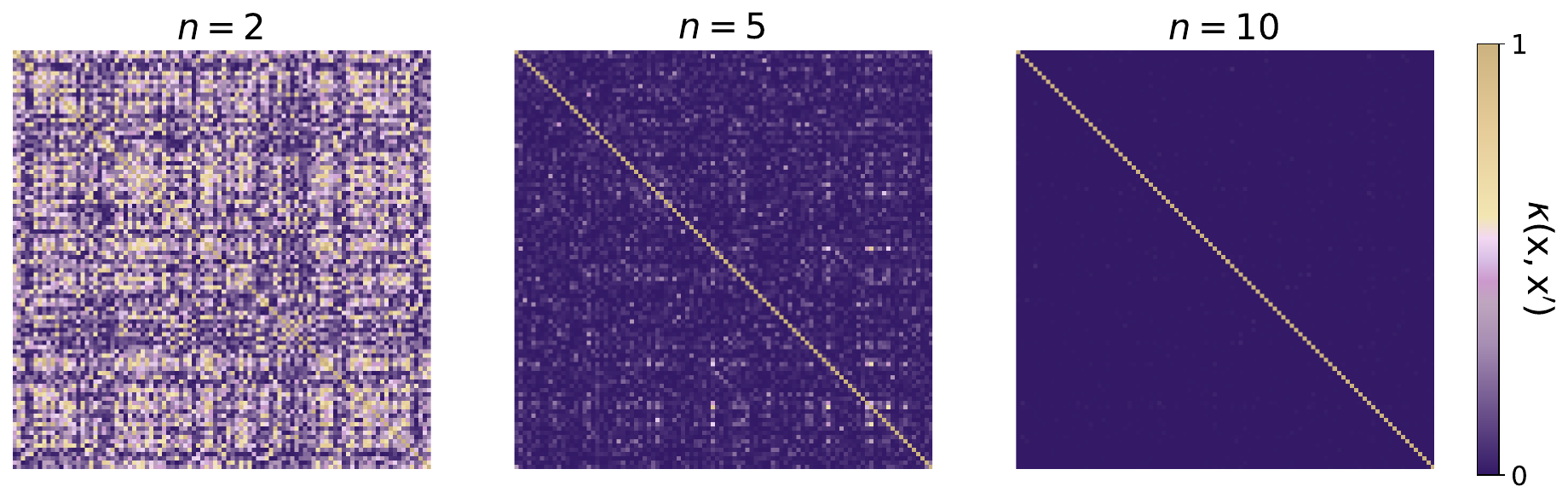}
  \caption{Exponential concentration of a quantum kernel: comparison of the kernel matrices constructed using an ansatz proposed by Havlíček  \textit{et al.} \cite{havlicek_supervised_2019} for binary classification on MNIST dataset \cite{mnist}. The matrices are constructed for 100 data points with the number of dimensions reduced to the corresponding number of qubits by principal component analysis. Each dimension is encoded into a single independent qubit. The kernel matrices shown are for qubit numbers $n = 2,5,10$.}
  \label{fig: eckernel}
\end{figure}

%\subsection{Quantum Neural Networks}
%\label{sec: qnn_theory}
The present work establishes the connection between EC of quantum kernels and BPs through properties of the cost function of a quantum neural network (QNN).
QNNs are a family of parameterized circuits that have been used for regression and classification tasks in quantum machine learning \cite{dymtroQNN_beer_training_2020, QNN_review_jeswal_recent_2019}. As shown in Fig.~\ref{fig: 1}, they can be defined using a data-based embedding $U_\mathbf{x}$ for $\mathbf{x} \in \mathcal{X}$ followed by variational embedding $W(\boldsymbol{\theta})$ with $m$ parameters $\boldsymbol{\theta} \in [ -\pi,\pi ]^m$.
%Usually $W(\theta)$ consists of multiple layers of ansatz, which increases the depth of the circuits and controls its expressivity.
The most general cost function for QNNs can be defined as follows:
\begin{equation}
\begin{split}
\mathcal{C}(\boldsymbol{\theta},\mathcal{X}) &= \sum_{i=1}^{N_s} c_{i}\text{Tr}[W(\boldsymbol{\theta})\rho(\mathbf{x}_i)W^{\dagger}(\boldsymbol{\theta})O] \\
&= \sum_{i=1}^{N_s} c_{i}\text{Tr}[W(\boldsymbol{\theta}) U(\mathbf{x}_i) \rho_0 U^\dagger(\mathbf{x}_i) W^{\dagger}(\boldsymbol{\theta})O ] = \sum_{i=1}^{N_s} c_{i}~C(\boldsymbol{\theta} , \mathbf{x}_i)
\end{split}
\label{eq:qnn_cost}
\end{equation}
where $O$ is a measurement operator. In this work, $O$ is a global measurement operator, acting on all qubits in the circuit in order to relate the cost function to quantum state overlaps serving as the kernel function.

%\subsubsection{Barren plateaus and cost concentration}
BPs are related  to concentration of the QNN cost function with respect to variational parameters of QNN. Mathematically, one can formulate BPs in terms of variance (and mean) of the partial derivatives of the cost function, as follows.
\begin{definition}
    \rm {\bf (Barren Plateaus)}. The cost function defined for QNN in Eq.~(\ref{eq:qnn_cost}) exhibits a BP if the variance of the derivative of the cost function with respect to $\theta_{\mu} \in \boldsymbol{\theta}$ is exponentially decaying in the number of qubits:
    \begin{equation}
        \mathrm{Var}_{\boldsymbol{\theta}}[\partial_\mu C(\boldsymbol{\theta},\mathbf{x}_i)] \leq F(n), \quad  F(n) \in O(1/b^n)
    \end{equation}
\end{definition}
%We can prove the mean of the derivative i.e. $\mathbb{E}_{\theta}[\partial_\mu C(\theta)]$  to be equal to zero in two ways. Intuitively, since the variational space is periodic, the integral of $\partial_\mu C(\theta)$ must vanish along a closed loop.

BPs were first introduced in Ref.~\cite{mcclean_barren_2018}, where the source of the exponential suppression was identified to be the high expressivity of quantum circuits, which was termed as randomness-induced BP. This affects both gradient-free \cite{grad_free_BP} and gradient-based \cite{grad_based_BP} optimization approaches. Specifically, it was argued that random initialization of $\boldsymbol{\theta}$ for a given $W(\boldsymbol{\theta})$ such that it forms a 2-design \cite{2design1, 2design2,mcclean_barren_2018} leads to a BP. Cerezo \textit{et al.} \cite{cerezo_cost_2021} also studied BPs in shallow circuits and concluded that the quantum circuit with global measurement-based cost functions ($O$ acting on all qubits) can suffer from BPs regardless of depth. The reason behind this suppression is that a global cost function compares two states in an exponentially large Hilbert space, which may lead to an exponentially small overlap, unidentifiable given that the number of measurements available scales only as ${\rm poly}(n)$.
While employing local cost functions in shallow circuits offers a BP mitigation strategy \cite{cerezo_cost_2021}, this approach risks eliminating potential quantum advantage. Specifically, since circuits limited to $\mathcal{O}(\log(n))$ layers are often classically simulable in polynomial time using tensor network techniques \cite{jozsa_SimulationQuantumCircuits_2006,leone_PracticalUsefulnessHardware_2022,cerezo_DoesProvableAbsence_2023}.
BPs have also been identified in the context of tensor network-based machine-learning models \cite{BP_tensor} and quantum circuits with high amounts of entanglement \cite{BP_ent}.
%A recent review \cite{larocca2024review} summarizes all the major developments in BPs over the past few years.
BPs are closely related to narrow gorges in a cost function landscape formed when the volumetric fraction of the parameter space below a certain value is exponentially suppressed. Rigorously, \emph{cost concentration} is defined according to Ref.~\cite{arrasmith2022equivalence} as follows.
\begin{definition}
    \rm {\bf (Cost concentration \& narrow gorges)}. The cost function defined for QNNs in Eq.~(\ref{eq:qnn_cost}) exhibits cost concentration if
    \begin{equation}
        \mathrm{Var}_{\boldsymbol{\theta}_{A}}[\mathbb{E}_{\boldsymbol{\theta}} [C(\boldsymbol{\theta},\mathbf{x}_i)]-C(\boldsymbol{\theta}_{A},\mathbf{x}_i)] \leq G(n),\quad  G(n) \in O(1/b^n), \quad b > 1.
\label{cc}
    \end{equation}
    where $\boldsymbol{\theta}_A$ is a random draw from the uniform probability distribution over the parameter space.
    Furthermore, the cost function defined for QNNs in Eq.~(\ref{eq:qnn_cost}) exhibits a narrow gorge if there exists a local minimum lower than the mean cost by at least $\delta(n)>0$ with $\delta(n) \in \Omega(1/{\rm poly}(n))$. The probability that the cost function value deviates from its mean by at least $\delta$ at a randomly chosen point from a uniform distribution over the parameter space is bounded and
    \begin{equation}
        \mathrm{Pr}_{\boldsymbol{\theta}_A}[ \mid \mathbb{E}_{\boldsymbol{\theta}} [C(\boldsymbol{\theta},\mathbf{x}_i)]-C(\boldsymbol{\theta}_{A},\mathbf{x}_i) \mid \geq \delta ] \mid \leq  \frac{G(n)}{\delta^2},\quad G(n) \in O(1/b^n), \quad b > 1
    \end{equation}
\end{definition}
\textcolor{black}{Importantly, Ref.~\cite{arrasmith2022equivalence} proves that, for parameterized quantum circuit cost landscapes,  exponentially vanishing gradients and exponential cost concentration about the mean, occur together. Therefore, the gradient-based BP definition above can be equivalently reformulated as cost concentration of the form in Eq.~\eqref{cc} and has become a widely adopted standard formulation of BPs \cite{larocca2024review}. In the remainder of this work, and in particular in Theorem~\ref{theorem1}, we use this cost-concentration formulation of BPs. This choice is natural for our purpose because exponential concentration of a quantum kernel is also a statement about the variance of a function value, namely the kernel entry \(\kappa(\mathbf{x},\mathbf{x}')\), rather than about gradients.} Note that the notation  $f(n)\in \mathcal{O}(g(n))$ and $f(n)\in \Omega(g(n))$ specifies that $f(n)$ is asymptotically bounded above and below by $g(n)$, respectively.

\section{Barren plateaus vs exponential concentration }
\label{sec: qml_results}
We relate EC in non-parametric (no trainable parameters) quantum kernels with the BP problem in QNNs. Thus, we are concerned with the exponential concentration of the actual kernel function and, hence, kernel matrix values.
% Both problems originate from a common source which is, taking the overlap of two states in an exponentially large Hilbert space. We want to find whether there is a more fundamental relation to this analogy, does one come from the other? In other words, can you use results in a barren plateau problem and translate it into a kernel formulation of the problem?
%For the case of parameterized quantum kernels, kernel target alignment can be considered as equivalent to the cost function in QNN \cite{qkreal2_hubregtsen2022training} and  the exponential upper bounds on the variance of the QNN cost function also limit the training of the kernel model \cite{thanasilp_exponential_2022}.
%Our goals is to demonstrate that a quantum ansatz proposed to mitigate BP eliminates EC of quantum kernels. We also show that an ansatz suffering from BP also leads to EC if used to build a quantum kernel.
We begin with a QNN with a BP. We consider the cost function given by Eq.~(\ref{eq:qnn_cost}) with fixed $\mathbf{x}_i$. Following Eq. (\ref{cc}), we can write
\begin{align}
    &\forall \mathbf{x}_i \in\mathcal{X}:\quad \mathrm{Var}_{\boldsymbol{\theta}_{A}}[\mu_{\mathbf{x}_i}-C(\boldsymbol{\theta}_{A},\mathbf{x}_i)] \leq F(n), \quad  F(n) \in \mathcal{O}\left(\frac{1}{b^n}\right) \label{eq: qNN_BP_bound}
\end{align}
where  $\mu_{\mathbf{x}_i}=\mathbb{E}_{\boldsymbol{\theta}} [C(\boldsymbol{\theta},\mathbf{x}_i)]$.
We propose the following construction of a quantum kernel: first, we replace the variational part of the circuit $W(\boldsymbol{\theta})$ with the data embedding $U^\dagger(\textbf{x}_i)$\footnote{Without loss of generality, the data $\mathcal{X}$ is assumed to be normalized and standardized to the range $[-\pi,\pi]$, which matches the typical range of parameters $\boldsymbol{\theta}$. \textcolor{black}{Importantly, this algebraic substitution assumes the data-induced unitary ensemble shares the concentration properties of the original parameter space.}}, and, second, we choose the measurement operator $O$ to be the initial state $\rho_0$. In this way, we recover the state preparation $\rho(\textbf{x})=W(\textbf{x})\rho_0 W^{\dagger}(\textbf{x})$ of a data point  $\textbf{x}' \in \mathcal{X}$ through the variational part of the QNN. This leads to the quantum kernel in the form
\begin{equation}
    \kappa(\textbf{x},\textbf{x}^{\prime})=\text{Tr}[W(\textbf{x})\rho_0 W^{\dagger}(\textbf{x})W(\textbf{x}^{\prime})\rho_0 W^{\dagger}(\textbf{x}^{\prime})].
    \label{eq: new_kernel}
\end{equation}
We next relate the exponential (polynomial) upper (lower) bound on the cost function of the QNN to EC of the corresponding kernels. To do this, we state the following lemmas with proofs provided in Appendix \ref{app:proofs_lemmas}.

\begin{lemma}\label{lemma:total_var_upper} \rm
   If there exists an upper bound under variation of the random variable $\mathbf{x}_j$,
    \begin{equation}
        \forall \mathbf{x}_i\!\!:\quad \operatorname{Var}_{\mathbf{x}_j}\left[f\left(\mathbf{x}_i, \mathbf{x}_j\right) \mid \mathbf{x}_i \right] \leq F(n)
    \end{equation}
    with function $f$ being symmetric in $\mathbf{x}_i, \mathbf{x}_j$,
    then this upper bound is doubled for the total variance:
    \begin{equation}
        \operatorname{Var}_{\mathbf{x}_i, \mathbf{x}_j}\left[f\left(\mathbf{x}_i, \mathbf{x}_j\right) \right] \leq 2F(n)
    \end{equation}
\end{lemma}
\begin{lemma}\label{lemma:total_var_lower} \rm
     If there exists a lower bound under variation of the random variable $\mathbf{x}_j$,
    \begin{equation}
        \forall \mathbf{x}_i\!\!:\quad \operatorname{Var}_{\mathbf{x}_j}\left[f\left(\mathbf{x}_i, \mathbf{x}_j\right) \mid \mathbf{x}_i \right] \geq G(n)
    \end{equation}
    then this lower bound also holds for the total variance:
    \begin{equation}
        \operatorname{Var}_{\mathbf{x}_i, \mathbf{x}_j}\left[f\left(\mathbf{x}_i, \mathbf{x}_j\right) \right] \geq G(n)
    \end{equation}
\end{lemma}

\noindent
 Given Lemmas \ref{lemma:total_var_upper} and \ref{lemma:total_var_lower}, we state the theorem, which relates the bounds for QNN to the bounds for the corresponding quantum kernel, with a proof provided in Appendix \ref{app:proof_main_theorem}.

\begin{theorem}\label{theorem1} \rm
    Given a QNN as defined in Eq.~(\ref{eq:qnn_cost}) that follows either an asymptotic exponential upper bound (BP landscape) or polynomial lower bound (BP-free landscape), there exists a quantum kernel as defined by Eq.~(\ref{eq: new_kernel}) such that the variance bounds transfer to the quantum kernel variance.
    Formally, for the QNN (Eq.~(\ref{eq:qnn_cost})) exhibiting a BP landscape, i.e.,
    \begin{equation}\label{eq: up_bound}
        \exists F(n) \forall \mathbf{x}\!: \quad
        \mathrm{Var}_{\boldsymbol{\theta}}[C(\boldsymbol{\theta},\mathbf{x})] \leq F(n)
        \quad \text{with} \quad
        F(n) \in \mathcal{O}\left(\frac{1}{b^n}\right)\!,\; b > 1,
    \end{equation}
    the corresponding quantum kernel (Eq.~(\ref{eq: new_kernel})) concentrates exponentially
    \begin{equation}
        \mathrm{Var}_{\textbf{x},\textbf{x}^{\prime}}[\kappa(\textbf{x},\textbf{x}^{\prime})] \leq 2F(n).
    \end{equation}
    Vice versa, if the QNN cost function (Eq.~(\ref{eq:qnn_cost})) is BP-free, i.e.,
    \begin{equation}\label{eq: lw_bound}
        \exists F(n) \forall \mathbf{x}\!: \quad
        \mathrm{Var}_{\boldsymbol{\theta}}[C(\boldsymbol{\theta},\mathbf{x})] \geq G(n)
        \quad \text{with} \quad
        G(n) \in \Omega\left(\frac{1}{\mathrm{poly}(n)}\right),
    \end{equation}
    the corresponding quantum kernel (Eq.~(\ref{eq: new_kernel})) does not concentrate exponentially
    \begin{equation}
        \mathrm{Var}_{\textbf{x},\textbf{x}^{\prime}}[\kappa(\textbf{x},\textbf{x}^{\prime})] \geq G(n)
        .
    \end{equation}
\end{theorem}
\noindent
Note that this theorem establishes equivalence by also covering the reverse relation: the two implications are logically equivalent to their contrapositives, i.e., the absence of EC implies the absence of BPs, and EC implies BPs, respectively. Consequently, for any QNN and quantum kernel corresponding to each other according to our framework, they will either both suffer from BPs and EC or both be free from BPs and EC. \textcolor{black}{We emphasize that Theorem~\ref{theorem1} is not a statement about arbitrary quantum kernels and arbitrary QNNs. The equivalence of concentration bounds holds for the kernel family constructed from the same ansatz according to Eq.~\eqref{eq: new_kernel}, with the observable choice $O=\rho_0$ and with the variance assumptions stated in Eqs.~\eqref{eq: up_bound} and \eqref{eq: lw_bound}. Thus, the theorem provides a recipe for transferring BP or BP-free guarantees from a QNN ansatz to its associated non-parametric quantum kernel, rather than a global equivalence between the two phenomena for unrelated circuit families.}

A physical intuition for this correspondence between BPs and EC lies in the geometry of the underlying vector spaces, following the picture in Cerezo \textit{et al.} \cite{cerezo_DoesProvableAbsence_2023}. In both cases, we effectively scatter objects randomly within a vector space and compare them via inner products, which reflect their similarity. The relevant vector space in this picture is the space of linear operators. For QNNs, both the input state $\rho(\bm{x})$ and the Heisenberg-evolved observable $W^\dagger(\bm{\theta}) O W(\bm{\theta})$ reside in this vector space. Analogously, for quantum fidelity kernels, the inner product is formed between two quantum states, $\rho(\bm{x})$ and $\rho(\bm{x}')$, rather than state and observable operators.
When the dimension of this vector space is exponentially large, the `curse of dimensionality' dictates that these inner products become exponentially small on average, leading simultaneously to BPs and EC. Conversely, restricting the quantum circuit or kernel construction confines this scattering to a polynomially-sized subspace, thereby ensuring that the inner products remain polynomially bounded.

One can argue that the choice of a state as the observable, $O = \rho_0$, is uncommon in typical literature on BPs, where observables are usually decomposed into Pauli measurements \cite{larocca2024review,invariant_schatzki_theoretical_2022,cerezo_cost_2021}
\begin{equation}\label{eq:obs_pauli_decomp}
    O = \sum_{k=1}^p c_k \left(\bigotimes_{i=1}^n \sigma_i^{(k)}\right) \quad \text{with}  \quad c_k \in \mathbb{R},\; \sigma_i^{(k)}\in \{I, X, Y, Z\},
\end{equation}
while we focus on computational basis measurements. However, it is important to note that computational basis states like $\rho_0$ can also be expressed in the Pauli operator basis. For example, on a single qubit, $|0\rangle\langle 0| = (I + Z) / 2$. Via the tensor product, this generalizes to $n$ qubits as
\begin{equation}
    \rho_0 = \sum_{\mathbf{b}\in \{0,1\}^n} \frac{1}{2^n}\bigotimes_{i=1}^n Z^{b_i} = O,
\end{equation}
matching the form of Eq.~\eqref{eq:obs_pauli_decomp}. Signs can be used to express any computational basis state.

\subsection{Numerical demonstration.}

\begin{figure}[h]
  \centering
  % Replace 'example-image' with the path to your image file
  \includegraphics[width=0.35\textwidth]{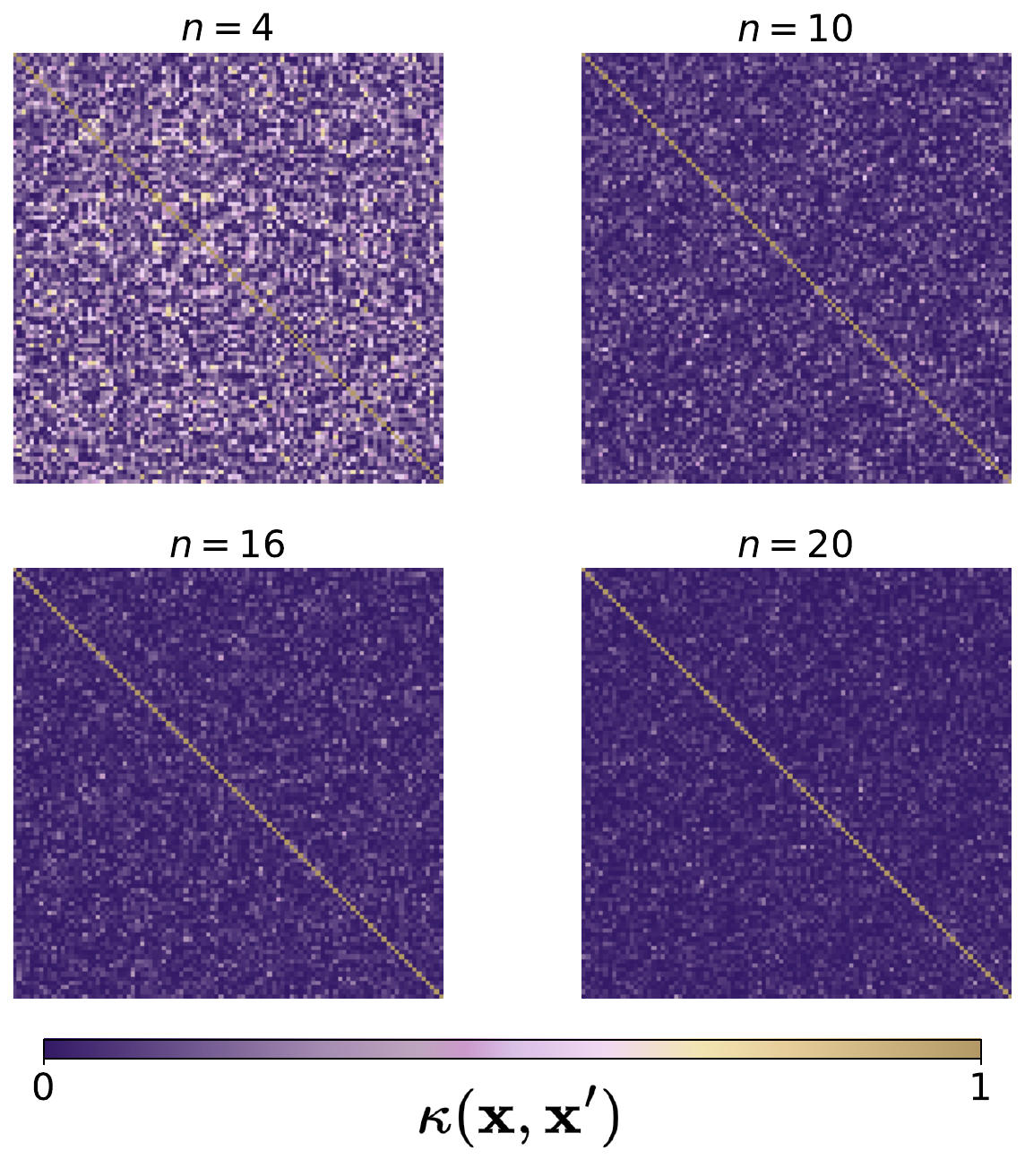}
  \hspace{3mm}
  {\includegraphics[width=0.475\textwidth]{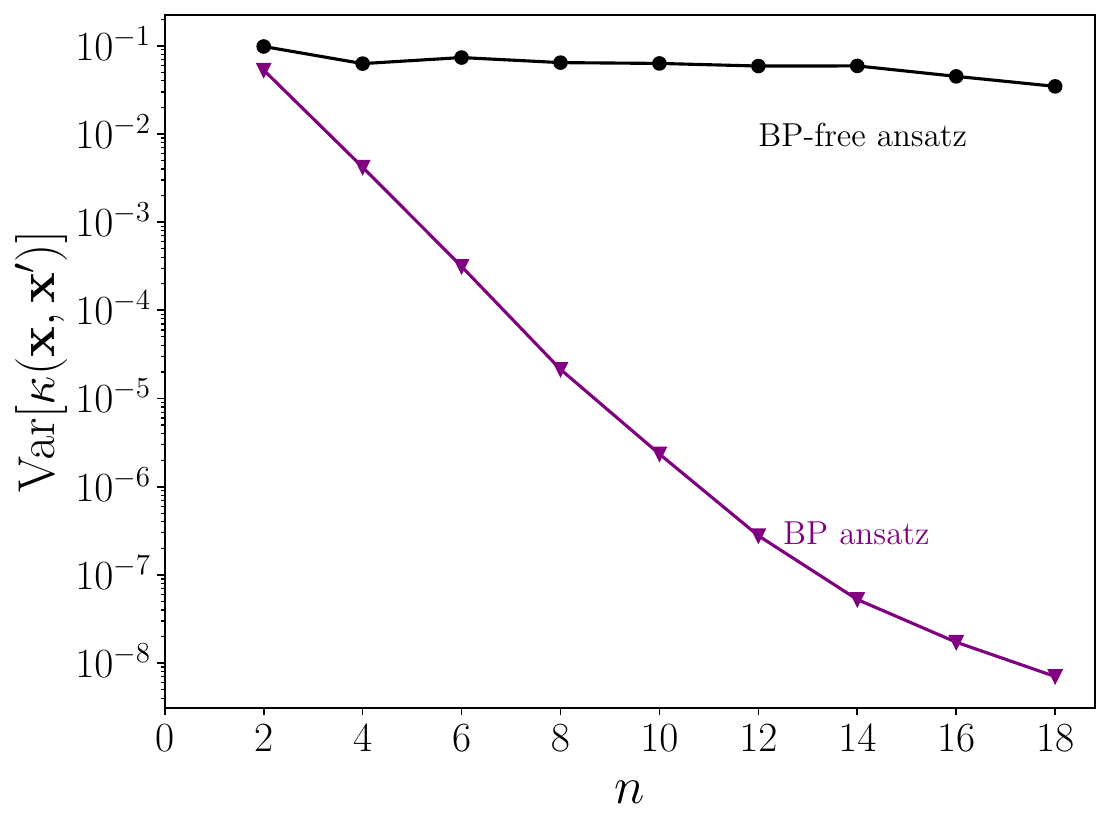}}

  \caption{Kernel matrices computed with the BP-free ansatz \cite{invariant_schatzki_theoretical_2022} for binary classification on the MNIST dataset \cite{mnist}.
  The matrices are constructed for 100 data points with the number of dimensions reduced to the corresponding number of qubits ($n$) by principal component analysis. Each dimension is encoded into a single independent qubit. (Right)  Variances of the kernel matrix elements computed with the BP-free ansatz (black) as a function of number of qubits in the range $n \in [2,20]$. The depth of the circuit is equal to the number of qubits at each point. The purple line corresponds to the Havlíček  \textit{et al.} \cite{havlicek_supervised_2019} considered in Fig. \ref{fig: 1}.
  % The depth of the circuit is kept constant.
  }
  \label{fig:3}
\end{figure}

To illustrate the results of Theorem \ref{theorem1} by numerical examples, we first consider the ansatz
 proposed in Ref. \cite{invariant_schatzki_theoretical_2022} which is provably BP-free. \textcolor{black}{We emphasize that these numerical examples serve an illustrative purpose; rather than acting as necessary empirical validation for our theory, they showcase how our framework can be used in practice to devise EC-free kernels from BP-free QNNs.} We use Eq.~(\ref{eq: new_kernel}) to construct a quantum kernel with this ansatz and evaluate the kernel matrices for the MNIST dataset considered before in Fig.~\ref{fig: eckernel}. Figure \ref{fig:3} shows the kernel matrices thus obtained for different numbers of qubits $n \in [5, 20]$. The 100 points are chosen randomly from the full dataset and not sorted by class. The kernel matrices in Fig.~\ref{fig:3} exhibit no concentration, as opposed to the kernel matrices in Fig.~\ref{fig: eckernel} where the off-diagonal elements of the kernel matrices become negligible already for $n \leq 10$. To provide a quantitative comparison between the results of Figs. \ref{fig: eckernel} and \ref{fig:3}, we compute the variance of $k({\bf x}, {\bf x}^\prime)$. As shown in the right panel of Fig.~\ref{fig:3}, ${\rm Var}_{{\bf x}, {\bf x}^\prime}[ k({\bf x}, {\bf x}^\prime)]$ computed with the BP-free ansatz decreases only moderately as the number of qubits increases, while the ansatz used for the results in Fig.~\ref{fig: eckernel} leads to an exponential collapse of the variance magnitude. For completeness, Fig.~\ref{fig:cerezo_ansatz} summarizes the two circuit families used for the present numerical calculations: the BP-free ansatz (with layer-wise parameter sharing and $ZZ$-generated two-qubit interactions) and the BP feature map built from single-qubit phase rotations together with CNOT entangling operations. \\

 To illustrate that the observed scaling is not specific to the MNIST dataset, we extend the kernel variance analysis to multiple datasets with qualitatively different intrinsic geometry. In particular, we use four datasets introduced in Ref.~\cite{bowles2024better}: hidden-manifold, two-curves, hyperplanes, and bars-and-stripes. For each dataset, we estimate the kernel variance ${\rm Var}_{{\bf x}, {\bf x}^\prime}[ k({\bf x}, {\bf x}^\prime)]$ as a function of the number of qubits $n$, using the same circuit constructions and scaling choices as in Fig. \ref{fig:cerezo_ansatz}. Specifically, we evaluate the BP-free permutation-equivariant ansatz of Ref.~\cite{invariant_schatzki_theoretical_2022} at depth equal to $n$, and use the feature map from Ref.~\cite{havlicek_supervised_2019} as a BP baseline for reference.

% \begin{figure}[ht]
% \centering
% \begin{minipage}[t]{0.99245\textwidth}
% \centering
% \includegraphics[width=\linewidth]{Equivalence_between_exponential_concentration_in_quantum_machine_learning_kernels_and_barren_plateaus_in_variational_quantum_algorithms/figures/fig_4a.pdf}

% \end{minipage}\hfill
% \begin{minipage}[t]{0.99245\textwidth}
% \centering
% \includegraphics[width=\linewidth]{Equivalence_between_exponential_concentration_in_quantum_machine_learning_kernels_and_barren_plateaus_in_variational_quantum_algorithms/figures/fig_4b.pdf}

% \end{minipage}
% \caption{Circuit ans\"atze used for the numerical calculations. Top panel: BP-free permutation-equivariant ansatz from Ref.~\cite{invariant_schatzki_theoretical_2022} (example shown for $L=3$ layers on $n=3$ qubits). Gates within a layer share a parameter $\theta_l$. Two-qubit gates are generated by $ZZ$ interactions. Bottom panel: BP ansatz from Ref.~\cite{havlicek_supervised_2019} (the $ZZ$ feature map in Qiskit), built from single-qubit phase rotations and entangling two-qubit CNOT operations.}
% \label{fig:cerezo_ansatz}
% \end{figure}
\begin{figure}[ht]
\centering
\begin{minipage}[t]{0.999245\textwidth}
\centering
% Added \fbox{} around the image to create a border
\fbox{\includegraphics[width=0.95\linewidth]{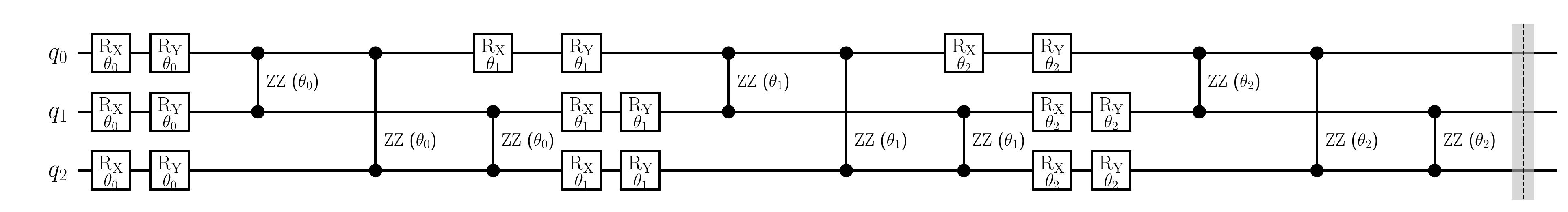}}
\end{minipage}\hfill

\vspace{0.5cm} % Optional: Adds a little breathing room between the two panels

\begin{minipage}[t]{0.95\textwidth}
\centering
% Added \fbox{} around the image to create a border
\fbox{\includegraphics[width=0.55\linewidth]{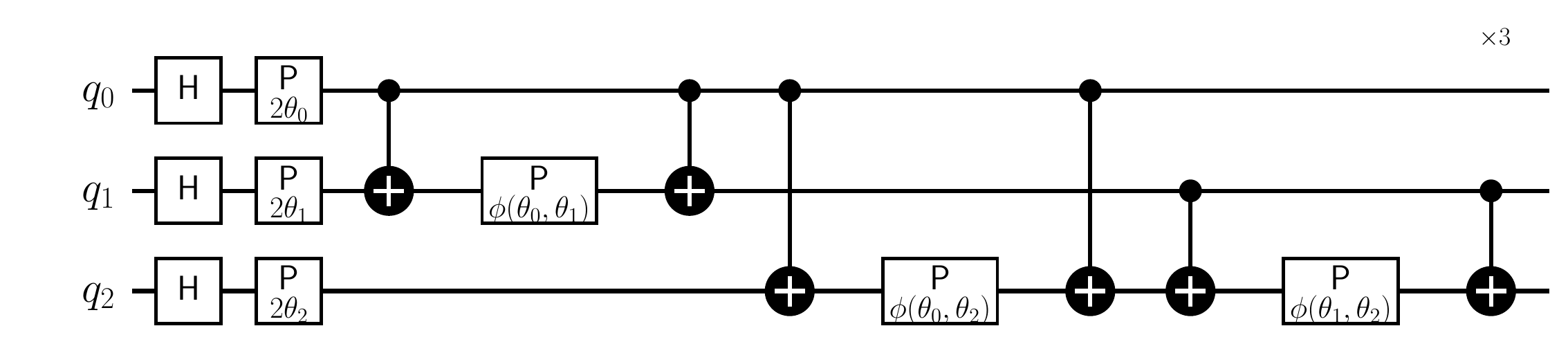}}
\end{minipage}
\caption{Circuit ans\"atze used for the numerical calculations. Top panel: BP-free permutation-equivariant ansatz from Ref.~\cite{invariant_schatzki_theoretical_2022} (example shown for $L=3$ layers on $n=3$ qubits). Gates within a layer share a parameter $\theta_l$. Two-qubit gates are generated by $ZZ$ interactions. Bottom panel: BP ansatz from Ref.~\cite{havlicek_supervised_2019} (the $ZZ$ feature map in Qiskit), built from single-qubit phase rotations and entangling two-qubit CNOT operations. $\phi$ represents the function $\phi(\theta_k,\theta_l)=2 (\pi-\theta_k)(\pi-\theta_l)$}
\label{fig:cerezo_ansatz}
\end{figure}
Figure~\ref{fig:multi_datasets} shows the resulting kernel variances for the four additional datasets. The qualitative behavior observed in Fig. \ref{fig:3} is consistent across all data sets, supporting the conclusion that the observed scaling is primarily a property of the circuit construction and depth scaling, rather than of a particular dataset. Finally, Fig.~\ref{fig:shot_noise} summarizes the impact of finite-shot sampling on the variance of kernel entries. We obtain each kernel entry by estimating the probability of the all-zero state via computational basis measurements, following the quantum kernel construction in Fig.~\ref{fig: 1}. Hence, each entry is obtained from $S$ repeated binary measurements as a sample proportion,
$\hat \kappa_{ij}=m_{ij}/S$ with $m_{ij}\sim\mathrm{Binomial}(S,\kappa_{ij})$, so that
$\mathrm{Var}(\hat \kappa_{ij}\,|\,\kappa_{ij})=\kappa_{ij}(1-\kappa_{ij})/S\le 1/(4S)$. Here $S$ denotes the total number of measurement shots to compute one kernel matrix element. Averaging over data point pairs $(i\neq j)$, this implies the following decomposition of the observed off-diagonal variance:
\begin{equation}
\mathrm{Var}_{i\neq j}[\hat \kappa_{ij}]
\;\approx\;
{\mathrm{Var}_{i\neq j}[\kappa_{ij}]}
\;+\;
{\mathbb{E}_{i\neq j}\!\left[\frac{\kappa_{ij}(1-\kappa_{ij})}{S}\right]},
\label{eq:shot_noise_var_decomp}
\end{equation}
up to correlation effects and finite-sample fluctuations.
The left panel of Fig. \ref{fig:shot_noise} compares $V_{\mathrm{ideal}}$ (statevector computation) to finite-shot estimates for several fixed shot budgets $S$. For the BP-free kernel, the finite-shot curves are nearly indistinguishable from the exact curve over the range shown, indicating that variance due to shot noise is negligible at these polynomial shot budgets. By contrast, for the $ZZ$ feature map the exact variance becomes extremely small with increasing $n$, and at low shot budgets the measured variance can be dominated by sampling noise, leading to a visible separation from the exact curve (e.g., already at $n=10$ for $S\lesssim 256$). The right panel provides direct evidence for Eq.~\eqref{eq:shot_noise_var_decomp} by comparing the measured finite-shot variance at fixed $n=10$ to the prediction $V_{\mathrm{pred}}(S)=V_{\mathrm{ideal}}+\mathbb{E}[\kappa(1-\kappa)]/S$. Note that there is no general requirement that the variance ``saturates'' at $1/S$ as $n$ increases: the effective shot-noise floor scales as $\mathbb{E}[\kappa_{ij}(1-\kappa_{ij})]/S$, and may itself decay with $n$ when typical overlaps $\kappa_{ij}$ decrease with system size.

\begin{figure}[ht]
\centering
\centering
\includegraphics[width=0.8\linewidth,height=0.8\linewidth,keepaspectratio]{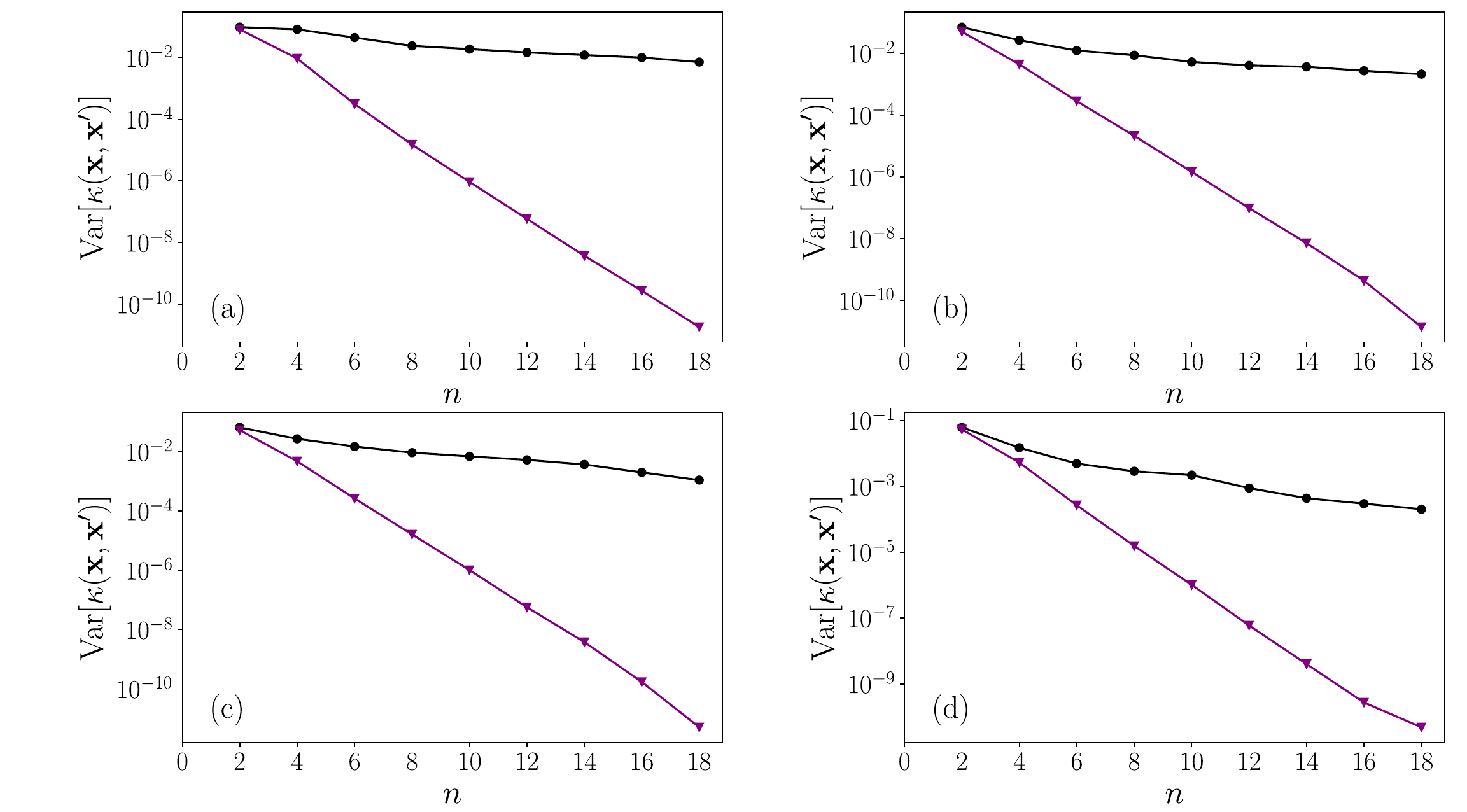}
\caption{Kernel-entry variances as functions of the qubit number for multiple datasets. We use the BP-free ansatz in Fig. \ref{fig:cerezo_ansatz} (circles) and set the circuit depth equal to $n$ for each calculation. The different panels correspond to different datasets from Ref.~\cite{bowles2024better}: (a) bars and stripes, (b) two curves, (c) hyperplanes, and (d) hidden manifold . The triangles are obtained with the feature map from Ref.~\cite{havlicek_supervised_2019} illustrated in Fig. \ref{fig:cerezo_ansatz} [top panel].}
\label{fig:multi_datasets}
\end{figure}

\begin{figure}[ht]
    \centering
    \includegraphics[width=0.5\linewidth]{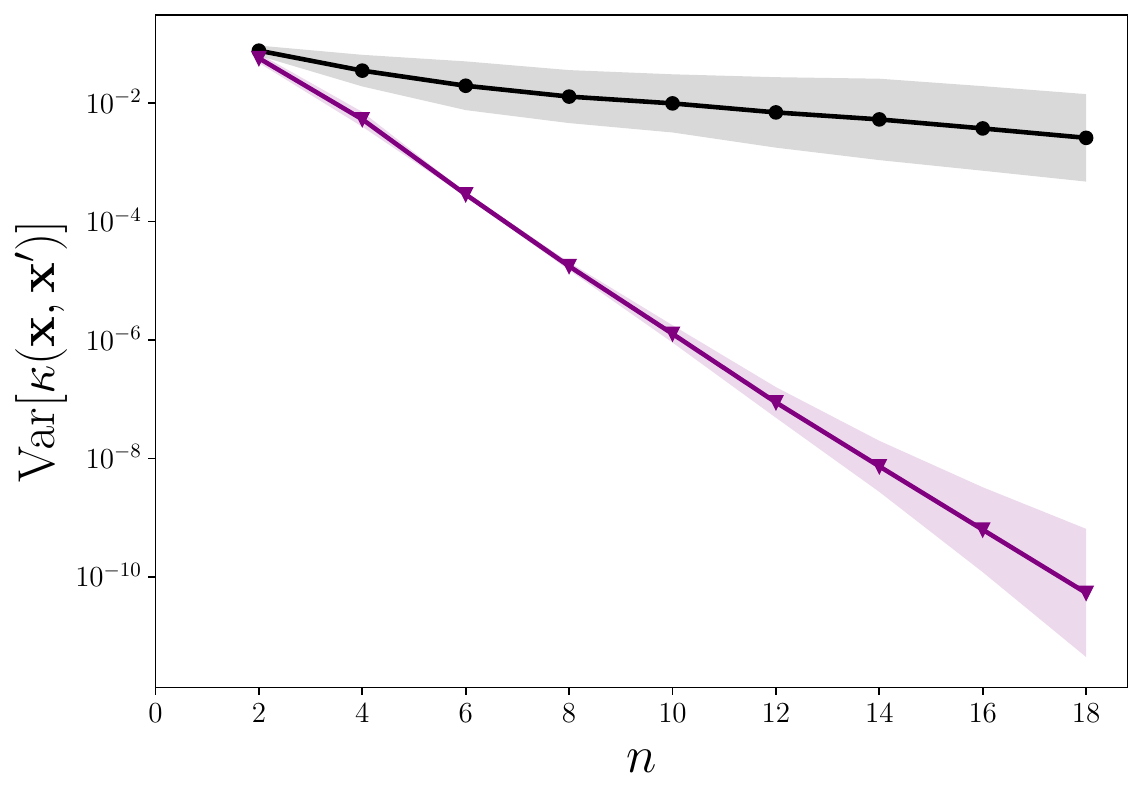}
    \caption{Kernel-entry variance averaged over the five benchmark datasets used in Figs.~\ref{fig: 1} and \ref{fig:multi_datasets}. \textcolor{black}{Error bars summarize the spread across heterogeneous benchmark datasets and are intended as descriptive, not inferential, intervals.}}
    \label{fig:var_mean_ci}
\end{figure}

\begin{figure}[ht]
\centering
\includegraphics[width=\linewidth]{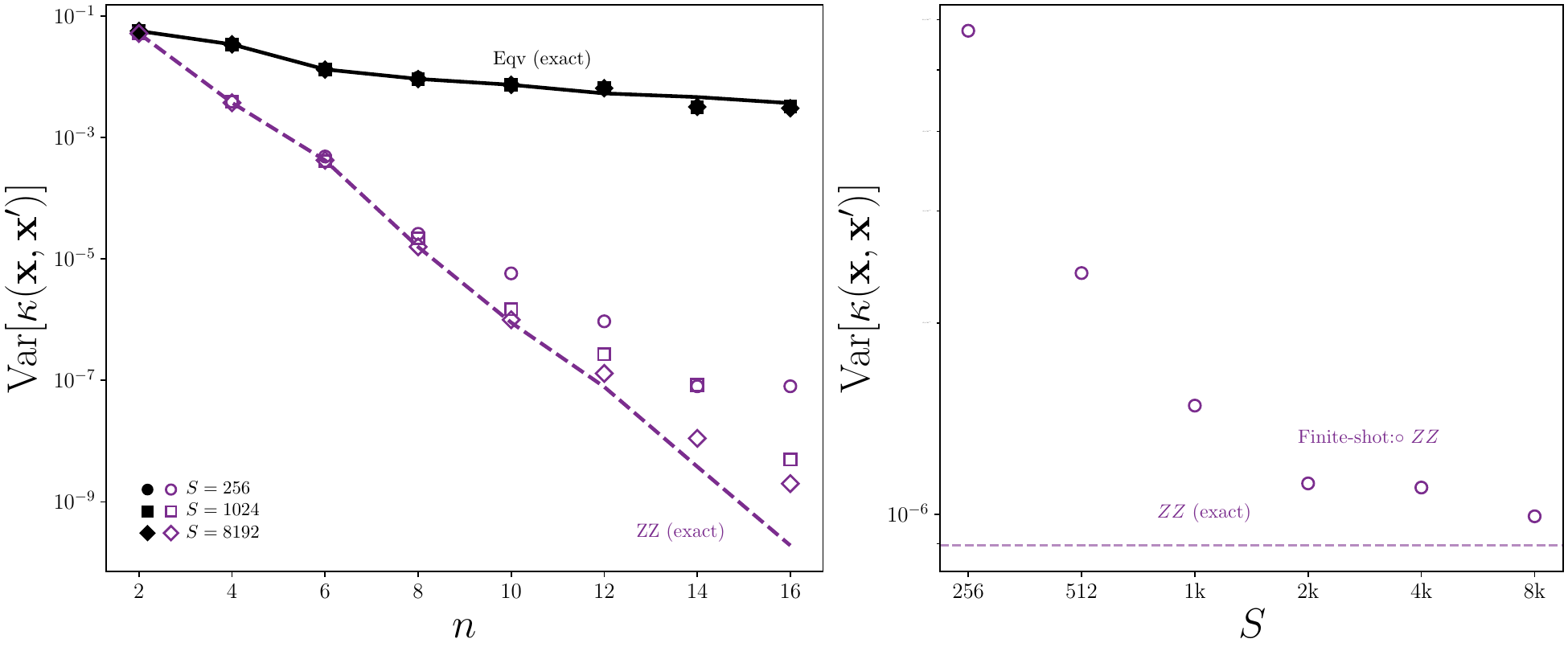}

\caption{Shot-noise effects on the kernel-entry variance. Left: Kernel-entry variances as functions of the qubit number for the BP-free kernel (full symbols) and the $ZZ$ feature map (open symbols), comparing exact (statevector) values to finite-shot estimates for several shot budgets $S\in\{256,\dots,8192\}$ describe each type of symbols. Right: validation of the shot-noise model at fixed $n=10$ by comparing measured finite-shot variances to the prediction $V_{\mathrm{pred}}(S)=V_{\mathrm{ideal}}+\mathbb{E}[\kappa(1-\kappa)]/S$.
}
\label{fig:shot_noise}
\end{figure}

 It is important to note that, despite apparent similarities, the EC-free kernel construction derived and analyzed here is not based on the principle of covariant quantum kernels \cite{glick_covariant_2022}. Covariant quantum kernels assume that the data are induced by an underlying group structure. In contrast, we derive permutation-invariant quantum kernels from permutation-invariant QNNs \cite{invariant_schatzki_theoretical_2022}, assuming that the data's semantic information remains unchanged under the action of the group. For instance, such an assumption applies if the input data vectors represent unordered sets of elements rather than lists where the order of elements is relevant and conveys information. Formally, when representing these sets via data vectors, any permutation of the components in retains the original meaning. Hence, the data is invariant under permutations realized through the action of the symmetric group.

\section{Conclusion}
We have shown that a quantum neural network can be used to build a quantum kernel that inherits the concentration properties of the QNN cost function.
More specifically, our results prove that the cost function concentration of a parametrized quantum circuit is equivalent to concentration of the (non-parametric) kernel matrix.
Our results imply that a QNN with a BP leads to quantum kernels that suffer from EC, whereas, more importantly, a BP-free QNN can be used to build a quantum kernel that does not exhibit exponential concentration.
We have illustrated this by a numerical example with an ansatz that was previously shown to be BP-free.
The present work implies that recently proposed algorithms for building BP-free quantum circuits can be applied to construct performant quantum kernels for machine learning with unknown inductive bias. Figures~\ref{fig:multi_datasets} and \ref{fig:shot_noise} highlight that the variance-scaling behavior is robust across datasets and is thus determined by the circuit and depth scaling, rather than the data geometry. Our numerical results also show that finite-shot estimates introduce an additive floor $\sim \mathbb{E}[K(1-K)]/S$, which is negligible for the BP-free kernel at polynomial $S$ but can dominate for the $ZZ$ feature map once the ideal variance becomes extremely small at large $n$.

It must be noted that constructing a quantum kernel from a BP-free QNN, which then does not suffer from exponential concentration, does not guarantee that the resulting quantum kernel is practically useful or has an advantage over classical kernels. We base this statement on recent evidence indicating that BP-free QNN constructions likely imply efficient classical simulability of the resulting quantum algorithm \cite{cerezo_DoesProvableAbsence_2023,leone_PracticalUsefulnessHardware_2022,basheer_AlternatingLayeredVariational_2022,Anschuetz2023efficientclassical,chang2026practicalframeworksimulatingpermutationequivariant}.
This is commonly referred to as \emph{de-quantization} of the quantum computing approach and makes the use of a quantum computer unnecessary.
This previous work and the connection between BP and EC proven here suggest that quantum kernels that do not exhibit EC are likely efficiently simulable with classical hardware.
Until a general BP-free QNN with a provable quantum advantage is found, the present work should be interpreted to suggest that quantum fidelity kernels do not offer a quantum advantage, unless they encode inductive bias. This is consistent with the conclusion in Ref. \cite{kubler2021inductive} based on the analysis of the eigenspectra of quantum kernels. This suggests that search for a quantum advantage of quantum kernel-based machine learning should focus on general strategies of encoding inductive bias.

%In the present work, we have shown that recently proposed algorithms for building BP-free quantum circuits can be applied to construct useful quantum kernels for machine learning without inductive bias. \textcolor{red}{Conclusion}

\begin{acknowledgments}
This work was supported by the Natural Sciences and
Engineering Research Council of Canada. We also acknowledge financial support from the Quantum Electronic
Science \& Technology (QuEST) Award given by the Stewart Blusson Quantum Matter Institute. We further acknowledge the NSERC CREATE in Quantum Computing Program, grant number 543245.
\end{acknowledgments}

\bibliography{ep_bc_bibfile}

\clearpage

\appendix

\section{Proofs}

    This appendix presents the proofs for the two lemmas and the main theorem of this work.

    \subsection{Proofs of lemmas}\label{app:proofs_lemmas}

        \setcounterref{lemma}{lemma:total_var_upper} % Set counter to the value of the referenced label
        \addtocounter{lemma}{-1}
        \begin{lemma}[\textit{restated}]\rm
            If there exists an upper bound under variation of the random variable $\mathbf{x}_j$,
            \begin{equation}
                \forall \mathbf{x}_i\!\!:\quad \operatorname{Var}_{\mathbf{x}_j}\left[f\left(\mathbf{x}_i, \mathbf{x}_j\right) \mid \mathbf{x}_i \right] \leq F(n)
            \end{equation}
            with function $f$ being symmetric in $\mathbf{x}_i, \mathbf{x}_j$,
            then this upper bound is doubled for the total variance:
            \begin{equation}
                \operatorname{Var}_{\mathbf{x}_i, \mathbf{x}_j}\left[f\left(\mathbf{x}_i, \mathbf{x}_j\right) \right] \leq 2F(n)
            \end{equation}
        \end{lemma}

        \begin{proof}
            The law of total variance provides the following decomposition
            \begin{equation} \label{eq:law_total_var}
                \operatorname{Var}_{\mathbf{x}_{i},{\mathbf{x}_j}}\left[f\left( \mathbf{x}_{i},\mathbf{x}_j\right)\right]
                = \mathbb{E}_{\mathbf{x}_i}\left[ \operatorname{Var}_{\mathbf{x}_j} \left[ f\left(\mathbf{x}_{i}, \mathbf{x}_j\right) |  {\mathbf{x}_i} \right] \right]
                +
                \operatorname{Var}_{\mathbf{x}_i}\left[ \mathbb{E}_{\mathbf{x}_j} \left[ f\left(\mathbf{x}_{i}, \mathbf{x}_j\right) |  {\mathbf{x}_i} \right] \right],
            \end{equation}
            where both terms are each upper bounded by $F(n)$. The bound on the first term trivially follows from the lemma's condition, and for the second term we decompose the variance into covariances~\cite{schervish2014probability}
            \begin{align}
                \operatorname{Var}_{\mathbf{x}_i}\left[ \mathbb{E}_{\mathbf{x}_j} \left[ f\!\left(\mathbf{x}_{i}, \mathbf{x}_j\right) |  {\mathbf{x}_i} \right] \right]
                &= \operatorname{Var}_{\mathbf{x}_i}\left[ \int_{\mathcal{X}} p(\mathbf{x}_j' | \mathbf{x}_i) f\!\left(\mathbf{x}_{i}, \mathbf{x}_j' \right) d\mathbf{x}_j' \right] \\
                &= \int_{\mathcal{X}}\int_{\mathcal{X}} p(\mathbf{x}_j' | \mathbf{x}_i)p(\mathbf{x}_j'' | \mathbf{x}_i) \operatorname{Cov}\!\left[f\!\left(\mathbf{x}_{i}, \mathbf{x}_j'\right), f\!\left(\mathbf{x}_{i}, \mathbf{x}_j''\right) \right] d\mathbf{x}_j' d\mathbf{x}_j'' ,
            \end{align}
            where the primes indicates realizations $\mathbf{x}'_j \in \mathcal{X}$ of the random variables $\mathbf{x}_j$ over $\mathcal{X}$.
            The {Cauchy-Schwarz inequality} provides upper bounds on the covariances \cite{schervish2014probability}
            \begin{equation}
                \sqrt{
                    \operatorname{Var}_{\mathbf{x}_i}\left[f\left(\mathbf{x}_i, \mathbf{x}_j'\right) \mid \mathbf{x}_j' \right]
                    \operatorname{Var}_{\mathbf{x}_i}\left[f\left(\mathbf{x}_i, \mathbf{x}_j''\right) \mid \mathbf{x}_j'' \right]
                },
            \end{equation}
            which, in turn, are upper bounded by $F(n)$ as per the lemma's condition and the symmetry of function $f$ in its arguments. Since $F(n)$ exists independently of realizations $\mathbf{x}_j', \mathbf{x}_j''$, the integrals trivially simplify, which concludes the proof by applying the derived bounds to Eq.~\eqref{eq:law_total_var}
            \begin{align}
                \operatorname{Var}_{\mathbf{x}_{i},{\mathbf{x}_j}}\left[f\left( \mathbf{x}_{i},\mathbf{x}_j\right)\right]
                &= \mathbb{E}_{\mathbf{x}_i}\left[ \operatorname{Var}_{\mathbf{x}_j} \left[ f\left(\mathbf{x}_{i}, \mathbf{x}_j\right) |  {\mathbf{x}_i} \right] \right]
                +
                \operatorname{Var}_{\mathbf{x}_i}\left[ \mathbb{E}_{\mathbf{x}_j} \left[ f\left(\mathbf{x}_{i}, \mathbf{x}_j\right) |  {\mathbf{x}_i} \right] \right] \\
                &\leq F(n) + F(n)\\
                &= 2 F(n)
            \end{align}
        \end{proof}

        \setcounterref{lemma}{lemma:total_var_lower} % Set counter to the value of the referenced label
        \addtocounter{lemma}{-1}
        \begin{lemma}[\textit{restated}] \rm
             If there exists a lower bound under variation of the random variable $\mathbf{x}_j$,
            \begin{equation}
                \forall \mathbf{x}_i\!\!:\quad \operatorname{Var}_{\mathbf{x}_j}\left[f\left(\mathbf{x}_i, \mathbf{x}_j\right) \mid \mathbf{x}_i \right] \geq G(n)
            \end{equation}
            then this lower bound also holds for the total variance:
            \begin{equation}
                \operatorname{Var}_{\mathbf{x}_i, \mathbf{x}_j}\left[f\left(\mathbf{x}_i, \mathbf{x}_j\right) \right] \geq G(n)
            \end{equation}
        \end{lemma}

        \begin{proof}
        \begin{align}
            \operatorname{Var}_{\mathbf{x}_{i},{\mathbf{x}_j}}\left[f\left(\mathbf{x}_{i}, \mathbf{x}_j\right)\right] &=\mathbb{E}_{\mathbf{x}_{i}, \mathbf{x}_j}\left[f\left(\mathbf{x}_{i,}, \mathbf{x}_j\right)^2\right]-\mathbb{E}_{\mathbf{x}_{i}, \mathbf{x}_j}\left[f\left(\mathbf{x}_i, \mathbf{x}_j\right)\right]^2 \\
            &=\mathbb{E}_{\mathbf{x}_i}\left[\mathbb{E}_{\mathbf{x}_j}\left[f\left(\mathbf{x}_i, \mathbf{x}_j\right)^2 \mid \mathbf{x}_i\right]\right]-\mathbb{E}_{\mathbf{x}_i}\left[\mathbb{E}_{\mathbf{x}_j}\left[f\left(\mathbf{x}_i, \mathbf{x}_j\right) \mid \mathbf{x}_i\right]\right]^2 \label{eq:law_total_expect}\\
            & \geq \mathbb{E}_{\mathbf{x}_i}\left[\mathbb{E}_{\mathbf{x}_j}\left[f\left(\mathbf{x}_i, \mathbf{x}_j\right)^2 \mid \mathbf{x}_i\right]\right]-\mathbb{E}_{\mathbf{x}_i}\left[\mathbb{E}_{\mathbf{x}_j}\left[f\left(\mathbf{x}_i, \mathbf{x}_j\right) \mid \mathbf{x}_i\right]^2\right] \label{eq:jensen}\\
            & =\mathbb{E}_{\mathbf{x}_i}\left[\mathbb{E}_{\mathbf{x}_j}\left[f\left(\mathbf{x}_i, \mathbf{x}_j\right)^2 \mid \mathbf{x}_i\right]-\mathbb{E}_{\mathbf{x}_j}\left[f\left(\mathbf{x}_{i,}, \mathbf{x}_j\right) \mid \mathbf{x}_i\right]^2\right] \\
            & =\mathbb{E}_{\mathbf{x}_i}\left[\operatorname{Var}_{\mathbf{x}_j}\left[f\left(\mathbf{x}_i, \mathbf{x}_j\right) \mid \mathbf{x}_i\right]\right] \\
            & \geq \mathbb{E}_{\mathbf{x}_i}[G(n)]=G(n).
        \end{align}
        The following properties are used: the law of total expectation is used in Eq.~\eqref{eq:law_total_expect}. Jensen's inequality, which states that the convex transformation of a mean is less than or equal to the mean applied after convex transformation, is applied in Eq.~\eqref{eq:jensen} as $x \mapsto x^2$ is a convex function. We also use linearity of expectation.
        \end{proof}

    \subsection{Proof of main theorem}\label{app:proof_main_theorem}

        \setcounterref{theorem}{theorem1} % Set counter to the value of the referenced label
        \addtocounter{theorem}{-1}
        \begin{theorem}[\textit{restated}]\rm
            Given a QNN as defined in Eq.~(\ref{eq:qnn_cost}) that follows either an asymptotic exponential upper bound (BP landscape) or polynomial lower bound (BP-free landscape), there exists a quantum kernel as defined by Eq.~(\ref{eq: new_kernel}) such that the variance bounds transfer to the quantum kernel variance.
            Formally, for the QNN (Eq.~(\ref{eq:qnn_cost})) exhibiting a BP landscape, i.e.,
            \begin{equation}
                \exists F(n) \forall \mathbf{x}\!: \quad
                \mathrm{Var}_{\boldsymbol{\theta}}[C(\boldsymbol{\theta},\mathbf{x})] \leq F(n)
                \quad \text{with} \quad
                F(n) \in \mathcal{O}\left(\frac{1}{b^n}\right)\!,\; b > 1,
            \end{equation}
            the corresponding quantum kernel (Eq.~(\ref{eq: new_kernel})) concentrates exponentially
            \begin{equation}
                \mathrm{Var}_{\textbf{x},\textbf{x}^{\prime}}[\kappa(\textbf{x},\textbf{x}^{\prime})] \leq 2F(n).
            \end{equation}
            Vice versa, if the QNN cost function (Eq.~(\ref{eq:qnn_cost})) is BP-free, i.e.,
            \begin{equation}
                \exists F(n) \forall \mathbf{x}\!: \quad
                \mathrm{Var}_{\boldsymbol{\theta}}[C(\boldsymbol{\theta},\mathbf{x})] \geq G(n)
                \quad \text{with} \quad
                G(n) \in \Omega\left(\frac{1}{\mathrm{poly}(n)}\right),
            \end{equation}
            the corresponding quantum kernel (Eq.~(\ref{eq: new_kernel})) does {not} concentrate exponentially
            \begin{equation}
                \mathrm{Var}_{\textbf{x},\textbf{x}^{\prime}}[\kappa(\textbf{x},\textbf{x}^{\prime})] \geq G(n)
                .
            \end{equation}
        \end{theorem}

        \begin{proof}
            The theorem assumes our construction introduced in the main text in Sec.~\ref{sec: qml_results}, summarized in the following conditions:
            \begin{itemize}
                \item same ansatz is chosen for the variational part of the QNN as for data embedding in the quantum kernel $U(\bm{x}) \leftrightarrow W^\dagger(\bm{x})$
                \item The initial state also serves as the observable $ O = \rho_0$
                % \item $\sigma = \rho_0$ and $\forall \bm{x}\!: M = U\!\left(\bm{x}\right) \rho_0 U^\dagger\!\left(\bm{x}\right)$, or $\forall \bm{x}\!:\sigma = U\!\left(\bm{x}\right) \rho_0 U^\dagger\!\left(\bm{x}\right) $ and $ O = \rho_0$
            \end{itemize}

        \noindent
            Under these conditions, Eq.~(\ref{eq: new_kernel}) applies, i.e., we recover the quantum kernel formulation from the QNN.
         %   one can transform the cost function for QNN given in Eq.~(\ref{eq: up_bound}) to an equivalent quantum kernel as defined in Eq.~(\ref{eq: new_kernel}):
           % \begin{equation}
             %   C(\theta | \mathbf{x}_i)=\mathrm{Tr}[W(\theta)O~W^{\dagger}(\theta) U(\mathbf{x}_i) \rho_0 U^{\dagger}(\mathbf{x}_i)] = \mathrm{Tr}[W(\theta)\rho_0~W^{\dagger}(\theta) W(\mathbf{x}_i) \rho_0 W^{\dagger}(\mathbf{x}_i)]
           % \end{equation}
            % \begin{equations}
            % \begin{align}
            %      C(\theta | \mathbf{x}_i)~=~&\mathrm{Tr}[W(\theta)O~W^{\dagger}(\theta) U(\mathbf{x}_i) \rho_0 U^{\dagger}(\mathbf{x}_i)] \\
            %      & = \mathrm{Tr}[W(\theta)\rho_0~W^{\dagger}(\theta) W(\mathbf{x}_i) \rho_0 W^{\dagger}(\mathbf{x}_i)] \\
            % \end{align}
            % \end{equations}
         Upon substitution $\boldsymbol{\theta} \rightarrow \mathbf{x}_j$, \textcolor{black}{which entails the assumption that the dataset distribution induces a unitary ensemble with analogous concentration properties to the given distribution over the parameter space,} Eq.~(\ref{eq: up_bound}) yields
            % \textcolor{red}{This is the part we need more clarification as to why we can do it, besides the fact that ranges for both are same. In theta space one can sample anywhere in the range but i don't know if that is true for data probability distribution}
            \begin{equation}
                \forall \mathbf{x}_i\!:\quad \operatorname{Var}_{\mathbf{x}_j}\left[\kappa \left(\mathbf{x}_i, \mathbf{x}_j\right) \mid \mathbf{x}_i \right] \in \mathcal{O}\left(\frac{1}{b^n}\right), \quad b > 1.
            \end{equation}
            Conversely, if no BPs exist using Eq.~(\ref{eq: lw_bound}) this implies a lower bound on the variance
            \begin{equation}
            \mathbf{x}_i\!:\quad \operatorname{Var}_{\mathbf{x}_j}\left[\kappa \left(\mathbf{x}_i, \mathbf{x}_j\right) \mid \mathbf{x}_i \right] \in \Omega\left(\frac{1}{\mathrm{poly}(n)}\right).
            \end{equation}
            Additionally, recall that by definition, the kernel function is symmetric in its arguments, i.e., $\kappa \left(\mathbf{x}_i, \mathbf{x}_j\right) = \kappa \left(\mathbf{x}_j, \mathbf{x}_i\right)$.
        %    These bounds apply to the total variance and, hence, kernel concentration as in Eqs. (\ref{eq: KC_upper}) corresponding to the BP in Eq.~(\ref{eq: up_bound}) can be proven using Lemma \ref{lemma:total_var_upper}. Similarly, for the BP-free situation as described in Eq.~\ref{eq: lw_bound}, we can obtain a lower bound on the variance of kernels as described in Eq.~(\ref{eq: KC_lower}) using Lemma \ref{lemma:total_var_lower}.
            It directly follows from Lemmas \ref{lemma:total_var_upper} and \ref{lemma:total_var_lower} that these asymptotic bounds also apply to the total variance (where constant factors are absorbed asymptotically)
            \begin{align}
            \operatorname{Var}_{\mathbf{x}_j,\mathbf{x}_i}\left[\kappa \left(\mathbf{x}_i, \mathbf{x}_j\right) \right] &\in \mathcal{O}\left(\frac{1}{b^n}\right) \quad \text{with} \quad b > 1 \quad \text{and} \quad \label{eq: KC_upper}\\
            \operatorname{Var}_{\mathbf{x}_j,\mathbf{x}_i}\left[\kappa \left(\mathbf{x}_i, \mathbf{x}_j\right) \right] &\in \Omega\left(\frac{1}{\mathrm{poly}(n)}\right), \label{eq: KC_lower}
            \end{align}
            respectively, which concludes the proof.
        \end{proof}

\end{document}